%% file: syncLCS.tex
\documentclass[11pt]{amsart}
\usepackage{amsmath}
\usepackage{amsfonts}
\usepackage[mathscr]{eucal}
\usepackage{amssymb}
\usepackage{bbm}
\allowdisplaybreaks

\newtheorem{thm}{Theorem}[section]

\theoremstyle{definition}
\newtheorem{defn}[thm]{Definition}

\newcommand{\bb}[1]{\mathbb{#1}}
\newcommand{\cl}[1]{\mathcal{#1}}

\begin{document}

\title[Synchronous LCS Games]{Synchronous Linear Constraint System Games}

\author[A.~Goldberg]{Adina Goldberg}
\thanks{This work was done with support from the Government of Ontario via an Ontario Graduate Scholarship.}
\address{Institute for Quantum Computing and Department of Pure Mathematics, University of Waterloo,
Waterloo, ON, Canada N2L 3G1}
\email{adina.goldberg@uwaterloo.ca}


\begin{abstract}
Synchronous linear constraint system games are nonlocal games that verify whether or not two players share a solution to a given system of equations. Two algebraic objects associated to these games encode information about the existence of perfect strategies. They are called the game algebra and the solution group. Here we show that these objects are essentially the same, i.e., that the game algebra is a suitable quotient of the group algebra of the solution group. We also demonstrate that linear constraint system games are equivalent to graph isomorphism games on a pair of graphs parameterized by the linear system.
\end{abstract}

\maketitle
\tableofcontents
\section{Introduction}
\input{sec/intro.tex}
\section{Synchronous games}
\label{sec:syncgames}
\input{sec/sync_games.tex}
\section{The syncLCS game}
\label{sec:syncLCS}

\input{sec/syncLCS_background.tex}
\section{The game algebra and the solution group}
\label{sec:solngp}
\input{sec/algebra_equiv.tex}

\input{sec/syncLCSproof.tex}
\section{Equivalence to the graph isomorphism game}
\label{sec:iso}

\input{sec/graph_games.tex}
\bibliographystyle{alpha}
\bibliography{syncLCS}

\end{document}

%% file: sec/intro.tex
Nonlocal games are recent objects of study in quantum information theory and complexity theory. They have been used to distinguish between various models of quantum entanglement. They have also appeared recently in operator algebras, playing a role in proving that the Connes embedding conjecture is false \cite{JNVWY}.

Here we unify and expand upon results about a nonlocal game called syncLCS$(A,b)$, parameterized by the system of equations $Ax=b$ over the finite field of order $p$. The game is meant to verify whether two players share a solution to the system of equations. Many tasks can be formulated in terms of solving a linear system of equations, so in a sense this is a general-purpose game. The game was introduced in \cite{KPS} as a synchronous version of the binary constraint system game from \cite{CM}.  Two distinct algebraic objects have been associated to this game, both encoding information about the existence of perfect strategies. One such object is a group called the solution group, due to \cite{CLS}, and the other is a *-algebra called the game algebra, due to \cite{OP} and described in more detail in \cite{HMPS}.

Our first result, Theorem \ref{thm:solngp}, states that the game algebra is *-isomorphic to a suitable quotient of the group algebra of the solution group. In other words, both algebraic objects encode the same information about the existence of perfect strategies.

Our second result, Theorem \ref{thm:sync-iso}, improves upon a theorem from \cite{BCEHPSW}. We show that the game syncLCS$(A,b)$ is *-equivalent to the graph isomorphism game, introduced in \cite{AMRSSV}, for a suitable pair of graphs. *-equivalence is useful for the following reason. When two games are *-equivalent, if there is a perfect strategy for one game in any of the standard quantum correlation sets (local, quantum tensor, quantum approximate, quantum commuting, etc.) then there is a perfect strategy for the other game in the same set.

This work builds on definitions in all of the papers cited above, but we restate them here for readability. Some familiarity with nonlocal games and with the classes of quantum correlations/strategies would be useful background for the reader, but is not strictly necessary in order to read the results or the proofs.

%% file: sec/sync_games.tex
Synchronous games were first hinted at in \cite{PSSTW} and properly introduced in \cite{HMPS}. We review the definition here. A synchronous game is a two player nonlocal game where both players have the same finite question (input) set, the same finite answer (output) set, and the players must give the same answer when receiving the same question. More formally\dots
\begin{defn}
A synchronous game $\cl G = (\cl I, \cl O, \lambda)$ is an ordered triple where $\cl I, \cl O$ are finite sets and $\lambda:\cl O^2 \times \cl I^2 \to \{0,1\}$, such that for $x \neq y\in \cl O, i\in \cl I$, $\lambda(x,y|i,i) = 0$.
\end{defn}

To a synchronous game we can associate a $*$-algebra that we will call the game algebra, capturing the structure of the game. This was first introduced in \cite{HMPS}. We review the game algebra construction here, as it is central to this paper.
\begin{defn}
Given a synchronous game $\cl G = (\cl I, \cl O, \lambda)$, the \emph{game algebra} $\cl A(\cl G)$ is the unital $*$-algebra over the complex numbers, generated by $e_{i,x}$ for $i\in \cl I, x\in\cl O$, with the following relations:
\begin{enumerate}
    \item $e_{i,x}^2 = e_{i,x} = e_{i,x}^*$
    \item \label{item:sumone} $\sum_{x\in\cl O} e_{i,x} = 1$
    \item \label{item:rules} $\lambda(x,y|i,j) = 0 \implies e_{i,x}e_{j,y} = 0$
\end{enumerate}
\end{defn}

Each projection $e_{i,x}$ corresponds to one of the players receiving the input $i$ and responding with output $x$. Relation \ref{item:sumone} captures the fact that a player must always respond with \emph{some} output. Relation \ref{item:rules} (which can actually be a whole collection of relations) captures the rules of the game, imposing orthogonality on input-output pairs that lose the game.

The game algebra captures a number of interesting properties about the existence of perfect strategies in various physical classes (corresponding to various quantum models). For example, there is a perfect deterministic strategy for the game $\cl G$ if and only if there is a unital $*$-homomorphism from $\cl A(\cl G)$ into $\mathbb{C}$. For proof, and for more of these correspondences, see \cite{HMPS}.

For the reader familiar with the idea of the $C^*$-algebra of a game as in \cite{LMPR}, that can be realized as a quotient of the game algebra.

We will need the following notion of equivalence of synchronous games.
\begin{defn}
We say two synchronous games $\cl G_1, \cl G_2$ are $*$-equivalent if there are unital $*$-homomorphisms $\cl A(\cl G_1) \to \cl A (\cl G_2)$ and $\cl A( \cl G_2) \to \cl A(\cl G_1)$.
\end{defn}
This notion of equivalence is notably quite a bit weaker than isomorphism of the game algebras. It is motivated by its preservation of perfect strategies in all of the relevant strategy classes.


%% file: sec/syncLCS_background.tex
Let $p$ be a prime number. Given an $m\times n$ matrix $A = (A_{i,j})$ and a vector $b = (b_i)$ over the field $\mathbb{Z}_p := \bb Z/p\bb Z$, the game $\textrm{syncLCS}(A,b)$ has Alice and Bob convincing a referee that they share a solution $x$ to the equation $Ax=b$. This game generalizes the syncBCS game introduced in \cite{KPS}, allowing more than just binary coefficients. LCS stands for `linear constraint system' whereas BCS stands for `binary constraint system', because in syncBCS, $p=2$.

For a vector $x\in Z_p^n$, let $\textrm{supp}(x) := \{j \in \{1,\dots,n\} : x_j \neq 0\}$ be the support of $x$, and let $f_i = [0 \dots 0\ 1\ 0 \dots 0]^T$ be the $i^\mathrm{th}$ standard basis vector in $Z_p^n$.  Notice that the $i^\mathrm{th}$ row of $A$ is given by $f_i^TA$.  For rows $i=1,\dots,m$, we let $$V_i := \textrm{supp}(f_i^T A).$$ So $V_i$ is the support of the $i^\mathrm{th}$ row of $A$. We then define $$S_i := \{ x \in \bb Z_p^n : f_i^T A x = b_i \textrm{ and } \textrm{supp}(x)\subseteq V_i\}.$$
$S_i$ is the set of solutions to the $i^\mathrm{th}$ equation, modulo discrepancies where coefficients of $A$ are $0$. When we are working with multiple games, there may be confusion over which linear system $S_i$ refers to. In that case, we will write $S_i(A,b)$ instead.

We now define the synchronous game associated to the system $Ax=b$ (the reader may verify that what follows is indeed a synchronous game).

\begin{defn}
    Let $A$ be an $m\times n$ matrix over $\bb Z_p$ and let $b \in \bb Z_p^n$. Then the \emph{synchronous LCS game} associated to the linear system $Ax=b$ and denoted $\mathrm{syncLCS}(A,b)$ is the tuple $(\cl I, \cl O, \lambda)$, where
    \begin{enumerate}
        \item the input set $\cl I = \{1,\dots,m\}$, corresponding to rows of $A$,
        \item the output set $\cl O = \bb Z_p^n$, corresponding to potential solutions to the system, and
        \item given inputs $(i,j)$ and outputs $(x,y)$, $\lambda(x,y|i,j)=1$ iff $(x,y)\in S_i \times S_j$ and $x_k = y_k$ for all $k\in V_i \cap V_j$.
    \end{enumerate}
\end{defn}

We can think of Alice and Bob winning the game if they provide compatible solutions to their respective equations. In fact, the notion of compatibility is quite useful in the remainder of this work, so we formalize it as follows.

\begin{defn}
Let $i,j\in\{1,\dots,m\}$, and $x\in S_i(A,b), y\in S_j(A,b)$. We say $(x,y)$ are \emph{compatible} solutions for rows $(i,j)$ if for all $k\in V_i\cap V_j$, $x_k = y_k$. Otherwise, we call the pair $(x,y)$ \emph{incompatible} for rows $(i,j)$.
\end{defn}

Note that when $i=j$, $(x,y)$ being compatible solutions to row $i$ is simply equivalent to $x=y$.

When $p=2$ in the syncLCS game, the game is referred to as syncBCS$(A,b)$, as introduced in \cite{KPS}. A nonsynchronous version of this game was introduced earlier in \cite{CLS}, where properties of syncBCS were shown to be captured by a group the authors called the solution group, denoted $\Gamma(A,b)$. The solution group was defined there for $p=2$. It was noted in that paper that the solution group could be defined for any prime $p$. Following that note, we set down a fairly direct generalization to the case when $p>2$, allowing us to associate a group to the syncLCS game.


\begin{defn}
\label{def:psolgp}
 The \emph{solution group} $\Gamma(A,b)$ is the group generated by symbols $g_1, \dots, g_n$ and $J$ satisfying the following relations.
\begin{enumerate}
    \item $g_j^p = 1$ for $j=1,\dots,n$
    \item $J^p = 1$
    \item $[g_j,J] = 1$ for $j=1,\dots,n$
    \item $[g_j,g_\ell] =1$ if $j,\ell\in V_i$ for some $i\in \{1,\dots,m\}$
    \item $$\prod_{j=1}^n g_j^{A_{i,j}} = J^{b_i} \text{ for } i=1,\dots,m $$
\end{enumerate}
\end{defn}

Each $g_j$ represents a variable $x_j$ in the system of equations $Ax = b$. $J$ is meant to represent a primitive $p^\textrm{th}$ root of unity. (In the case that $p=2$, the solution group can be thought of as a multiplicative representation of the system of equations, where $0$s become $1$s, and $1$s become $J$s, which should be thought of as $-1$s.) Any two variables that appear in the same equation in the system $Ax=b$ are required to have corresponding group elements that commute. Finally, the variables in a given equation should have their corresponding group elements multiply to an appropriate power of $J$ depending on the value of $b_i$. Note that there is a one-to-one correspondence between powers of $J$ and possible values for $b_i$.

%% file: sec/algebra_equiv.tex
Let $\omega \in \mathbb{C}$ be a primitive $p^\text{th}$ root of unity. Consider the $*$-algebra $\bb C \Gamma(A,b)$. In this algebra, there is an element $\omega$ and an element $J$. In keeping with the motivation for constructing the solution group, we identify $J$ with $\omega$ through a quotient. (In the $p=2$ case, this is the identification that ensures $J$ really does represent $-1$.) Similarly to $\cl A(\mathrm{syncLCS}(A,b))$, the resulting $*$-algebra encodes information about perfect strategies in the syncLCS game. In fact, the two algebras are isomorphic -- the first of the main results of this paper.

\begin{thm}
\label{thm:solngp}
The following algebras are isomorphic as unital $*$-algebras.
$$\cl A(\mathrm{syncLCS}(A,b)) \cong \frac{\bb C \Gamma(A,b)}{\langle J-\omega 1 \rangle}$$
\end{thm}

There is a related notion to the game algebra, called the $C^*$-algebra of a game, denoted $C^*(\cl G)$, that also encodes information about perfect strategies. $C^*(\cl G)$ is a quotient of the game algebra. Theorem \ref{thm:solngp} resembles \cite[Corollary 4.10]{LMPR} but is more general in that it is a result at the level of *-algebras as opposed to difficult-to-construct quotients. This result is also more general in that we do not restrict $p=2$. 

Before embarking on the proof of this theorem, let us briefly examine the algebra $A(\mathrm{syncLCS}(A,b))$. Following the game-algebra construction, we get a $*$-algebra over $\mathbb{C}$ generated by self-adjoint idempotents $a_{i,x}$ where $i\in\{1,\dots,m\}$ and $x\in S_i(A,b)$, satisfying the following relations:
\begin{enumerate}
    \item $\sum_{x\in S_i}a_{i,x} = 1$
    \item When $(x,y)$ are incompatible solutions to rows $(i,k)$ then $a_{i,x}a_{k,y} = 0$.
\end{enumerate}

%% file: sec/syncLCSproof.tex
\begin{proof}
We will define two $*$-algebra homomorphisms $\tilde\phi$ and $\overline{\psi}$ and then show that they are mutual inverses.

Let $a_{i,x}$ represent the generator of $\cl A(\mathrm{syncLCS}(A,b))$ corresponding to row $i$ and solution $x\in S_i$.

Define $ \phi : \bb C \Gamma(A,b) \to \cl A(\mathrm{syncLCS}(A,b))$ on the generators by
\begin{align*}
    \phi(g_j) &= \sum_{x\in S_i} \omega^{x_j}a_{i,x}\\
    \phi(J) &= \omega 1
\end{align*}
where $i$ is chosen so that $j\in V_i$.

We must check that $\phi$ is well-defined, i.e.\ that $\sum_{x\in S_i} \omega^{x_j}a_{i,x}$ doesn't depend on the choice of $i$.

Let $j \in V_i \cap V_k$. Fix $t \in \bb Z_p$.

Let $P_{i,j}(t):= \sum_{\substack{x\in S_i\\ x_j = t}}a_{i,x}$. Note that $\phi(g_j) = \sum_{t\in \bb Z_p} \omega^t P_{i,j}(t)$. By properties of the $a_{i,x}$:
\begin{align*}
    &P_{i,j}(t)^2 = P_{i,j}(t)&\\
    &P_{i,j}(t)P_{k,j}(s) = 0 & \text{for }s\neq t\\
    &\sum_{t\in \bb Z_p}P_{i,j}(t) = 1
\end{align*}
Therefore,
\begin{align*}
    P_{i,j}(t) &= P_{i,j}(t) \sum_{s\in \bb Z_p} P_{k,j}(s)\\
    &= P_{i,j}(t) P_{k,j}(t)\\
    &= \left(\sum_{t\in \bb Z_p} P_{i,j}(s)\right)P_{k,j}(t)\\
    &= P_{k,j}(t)
\end{align*}
But then we can write $\phi(g_j) = \sum_{t\in \bb Z_p} \omega^t P_{k,j}(t)$, so $\phi$ is well-defined.

We now show that $\phi$ is a $*$-algebra homomorphism by checking that the relations of $\Gamma(A,b)$ given in Definition \ref{def:psolgp} are satisfied, following the same numbering as in the definition.


\begin{enumerate}
\item We first check that $\phi(g_j)^p = 1$. Indeed,
        $$\phi(g_j)^p = \sum_{x\in S_i} \omega^{px_j}a_{i,x}^p = \sum_{x\in S_i}a_{i,x} = 1.$$
    as $a_{i,x}a_{i,y} = 0$ for $x\neq y$.
\item We now check that $\phi(J)^p = 1$.
Indeed, $$\phi(J)^p = \omega^p 1 = 1.$$ 
\item Each $\phi(g_j)$ commutes with $\phi(J) = \omega \in \bb{C}$.
\item Let $i\in\{1,\dots,m\}$. Let $j,\ell \in V_i$. We shall check that $\phi(g_j)$ commutes with $\phi(g_\ell)$. Recall that for $x,y \in S_i$, $a_{i,x}a_{i,y} = 0$ if $x\neq y$. So $\{a_{i,x}: x\in S_i\}$ commute. As $\phi(g_j)$ and $\phi(g_\ell)$ are linear combinations of this set, they commute as well.
\item Let $i \in \{1,\dots,m\}$. We must finally check that $\prod_{j\in V_i} \phi(g_j)^{A_{i,j}} = \phi(J)^{b_i}$.
\begin{align*}
    \prod_{j\in V_i} \phi(g_j)^{A_{i,j}} &=  \prod_{j\in V_i} \left(\sum_{x\in S_i} \omega^{A_{i,j}x_j}a_{i,x}\right)\\
    &= \sum_{x\in S_i}\left(\prod_{j\in V_i}\omega^{A_{i,j}x_j}\right)a_{i,x}\\
    &= \sum_{x\in S_i}\omega^{b_i}a_{i,x}\\
    &= \omega^{b_i}1\\
    &= \phi(J)^{b_i}
\end{align*}
\end{enumerate}

So $\phi: \bb C \Gamma(A,b) \to \cl A(\mathrm{syncLCS}(A,b))$ is a homomorphism of $*$-algebras. Because $J-\omega1 \in \textrm{ker}(\phi)$, there is an induced homomorphism $\frac{\bb C \Gamma(A,b)}{\langle J - \omega 1\rangle} \to \cl A(\mathrm{syncLCS}(A,b))$ which we will denote by $\tilde\phi$.



Now for the reverse direction we define 
\begin{align*}
    \psi : \cl A(\mathrm{syncLCS}(A,b)) &\to \bb C \Gamma(A,b) \\
    a_{i,x} &\mapsto \prod_{j\in V_i}f_j(x_j)
\end{align*}
where for $j\in\{1,\dots,m\}$, we have
\begin{align*}
    f_j:\bb Z_p &\to \bb C \Gamma(A,b)\\
    s &\mapsto \frac{1}{p} \sum_{t=0}^{p-1}(\omega^{-s}g_j)^t
\end{align*}

After checking $\psi$ is a $*$-homomorphism, we can compose it with the quotient map $\bb C \Gamma(A,b) \to \frac{\bb C \Gamma(A,b)}{\langle J - \omega 1\rangle}$ to get the $*$-homomorphism we are looking for:
$$\overline{\psi}: \cl A(\mathrm{syncLCS}(A,b)) \to \frac{\bb C \Gamma(A,b)}{\langle J - \omega 1\rangle} $$ 


In order to see that $\psi$ is a $*$-homomorphism, we will first demonstrate some facts about the family $f_j$. First and foremost, we will show that each $f_j(s)$ is a self-adjoint projection.

\begin{align*}
    f_j(s)^2
    &= \frac{1}{p^2}\left(1 + \omega^{-s} g_j + (\omega^{-s} g_j)^2 + \dots + (\omega^{-s} g_j)^{p-1}\right)^2\\
    &= \frac{p}{p^2} \left(1 + \omega^{-s} g_j + (\omega^{-s} g_j)^2 + \dots + (\omega^{-s} g_j)^{p-1} \right)\\
    &= f_j(s)\\
    f_j(s)^* &= \frac{1}{p} \sum_{t=0}^{p-1}(\omega^s g_j^{-1})^t\\
    &= \frac{1}{p}\sum_{r=1}^p \omega^{-s(r-p)}g_j^{r-p}\\
    &= \frac{1}{p} \sum_{r=0}^{p-1} (\omega^{-s}g_j)^r\\
    &= f_j(s)
\end{align*}

Now we will see that $g_j$ decomposes as a linear combination of the $f_j(s)$ as follows:
\begin{align*}
    \sum_{s=0}^{p-1} \omega^s f_j(s) &=\frac{1}{p} \sum_{s,t=0}^{p-1} \omega^{s-st}g_j^t\\
    &=\frac{1}{p} \sum_{t=0}^{p-1} g_j^t \sum_{s=0}^{p-1}\omega^{(1-t)s}\\
    &=\frac{1}{p} g_j^1 \sum_{s=0}^{p-1}\omega^0 + \frac{1}{p}  \sum_{r=1}^{p-1}g_j^{1-r} \sum_{s=0}^{p-1} (\omega^r)^s\\
    &= \frac{p}{p} g_j + \frac{1}{p}\sum_{r=1}^{p-1}g_j^{1-r} 0\\
    &= g_j
\end{align*}

Note further that if there is some $i\in\{1,\dots,m\}$ with $j,\ell \in V_i$ then $f_j(s)$ commutes with $f_\ell(r)$ for all $r,s$, precisely because $g_j$ commutes with $g_\ell$. 

Additionally, $\{f_j(s):s\in\bb Z_p\}$ form an orthogonal family summing to $1$. For $s\neq r$:
\begin{align*}
    f_j(s)f_j(r) &= \frac{1}{p^2}\sum_{t,u=0}^{p-1} \omega^{-st-ru}g_j^{t+u}\\
    &= \frac{1}{p^2}\sum_{\ell=0}^{p-1}\left(\sum_{\substack{t,u=0\\t+u=\ell}}^{p-1} \omega^{-st-ru}\right)g_j^\ell\\
    &=\frac{1}{p^2}\sum_{\ell=0}^{p-1}\left(\sum_{t=0}^{p-1} \omega^{-(s-r)t}\right)\omega^{-r\ell}g_j^\ell\\
    &= \frac{1}{p^2}\sum_{\ell=0}^{p-1}0\cdot\omega^{-r\ell}g_j^\ell\\
    &= 0\\
    \sum_{s=0}^{p-1} f_j(s) &= \frac{1}{p} \sum_{s,t=0}^{p-1}(\omega^{-s}g_j)^t\\
    &= \sum_{t=0}^{p-1}\left(\frac{1}{p}\sum_{s=0}^{p-1}\omega^{-st}\right)g_j^t\\
    &= \frac{p}{p}g_j^0 + \frac{1}{p}\left(1 + \omega + \omega^2 + \dots + \omega^{p-1}\right)\sum_{t=1}^{p-1}g_j^t\\
    &= 1
\end{align*}

With these properties of $f_j(s)$, it is fairly easy to extend them and prove all but one required property of $\psi$.

For each $a_{i,x}$, we have $\psi(a_{i,x})$ is a self-adjoint projection, as it is a commuting product of self-adjoint projections.

Now let $(x,y)$ be a pair of incompatible solutions to rows $(i,k)$. Then there exists $u\in V_i \cap V_k$ with $x_u \neq y_u$. But then $f_u(x_u)f_u(y_u) = 0$. So 
    \begin{align*}
        \psi(a_{i,x})\psi(a_{k,y}) &= \prod_{j\in V_i} f_j(x_j) \prod_{\ell \in V_k} f_\ell(y_\ell)\\
        &=\left(\prod_{j\in V_i\backslash \{u\}} f_j(x_j) \right)f_u(x_u)f_u(y_u)\left(\prod_{\ell \in V_k\backslash \{u\}} f_\ell(y_\ell)\right)\\
        &= 0
    \end{align*}
    \item The last property we need to show is that $\sum_{x\in S_i} \psi(a_{i,x}) = 1$. This is the trickiest one to prove and requires the use of some involved symmetry arguments. We will first note that because $A_{i,j}\neq 0 \in \bb Z_p$ for $j\in V_i$, we can apply a change of variables in our definition of $f_j$ as follows, replacing $t$ with $tA_{i,j}$:
    $$f_j(s) = \frac{1}{p}\sum_{t=0}^{p-1} (\omega^{-s}g_j)^{tA_{i,j}}$$
    We will let the set $\cl P_i$ denote the set of partitions $\lambda$ of the set $V_i$ into exactly $p$ blocks, where the $t^{\text{th}}$ block is denoted $\lambda_t$. More precisely, each $\lambda = \{\lambda_t : t = 0,\dots,p-1\}$, satisfying $V_i = \cup_{t} \lambda_t$ and $\lambda_t \cap \lambda_s = \emptyset$ when $s\neq t$. Note that $\lambda_t$ may be empty.
    
    We then have
    \begin{align*}
        p^{|V_i|}\sum_{x\in S_i} \psi(a_{i,x}) &= \sum_{x\in S_i} \prod_{j\in V_i} \sum_{t=0}^{p-1} (\omega^{-x_j}g_j)^{tA_{i,j}}\\
        &= \sum_{x\in S_i}\sum_{\lambda \in \cl P_i} \prod_{t=0}^{p-1}\prod_{j\in\lambda_t}(\omega^{-x_j}g_j)^{tA{i,j}}\\
        &= \sum_{\lambda \in \cl P_i}\sum_{x\in S_i}\prod_{t=0}^{p-1}\left( \omega^{-\sum_{j\in\lambda_t}A_{i,j}x_j}\prod_{j\in\lambda_t}g_j^{A_{i,j}}\right)^t
    \end{align*}
    We will divide the outer sum into a sum where $\lambda$ has exactly one non-empty block $\lambda_s = V_i$ and a sum where $\lambda$ has more than one non-empty block. We will denote the second set of partitions by $\Tilde{\cl P_i}$. The sum splits into the following two pieces.
    \begin{align}
        \label{eq:oneblock}\sum_{s=0}^{p-1}\sum_{x\in S_i} \left( \omega^{-\sum_{j\in V_i } A_{i,j}x_j} \prod_{j\in V_i} g_j^{A_{i,j}}\right)^s\\
        \label{eq:moreblocks}\sum_{\lambda \in \tilde{\cl P_i}}\sum_{x\in S_i}\prod_{t=0}^{p-1}\left( \omega^{-\sum_{j\in\lambda_t}A_{i,j}x_j}\prod_{j\in\lambda_t}g_j^{A_{i,j}}\right)^t
    \end{align}
    and we have that 
    $p^{|V_i|}\sum_{x\in S_i} \psi(a_{i,x})$ is the sum of \eqref{eq:oneblock} and \eqref{eq:moreblocks}.
    We will now show that $\eqref{eq:oneblock} = p^{|V_i|}$ and $\eqref{eq:moreblocks} = 0$, completing the proof of this particular property.
    \begin{align*}
        \eqref{eq:oneblock} &= \sum_{s=0}^{p-1}\sum_{x\in S_i} \left( \omega^{-b_i} \omega^{b_i}\cdot 1\right)^s\\
        &= \sum_{s=0}^{p-1}\sum_{x\in S_i}1\\
        &= p \cdot |S_i|\\
        &= p^{|V_i|}
    \end{align*}
    
    \begin{align*}
        \eqref{eq:moreblocks} &= \sum_{\lambda \in \tilde{\cl P_i}}\sum_{x\in S_i}\prod_{t=0}^{p-1} \omega^{ -t\sum_{j\in\lambda_t} A_{i,j} x_j} \prod_{j \in \lambda_t} g_j^{tA_{i,j}}\\
        &=  \sum_{\lambda \in \tilde{\cl P_i}}\left(\sum_{x\in S_i} \omega^{ -\sum_{t=0}^{p-1}\sum_{j\in\lambda_t} tA_{i,j} x_j} \right)\prod_{t=0}^{p-1} \prod_{j \in \lambda_t} g_j^{tA_{i,j}}\\
    \end{align*}
    For the following paragraph, fix $\lambda \in \tilde{\cl P_i}$. Let \begin{align*}
        \mu_i: \bb Z_p^n &\to \bb Z_p\\
        x &\mapsto -\sum_{t=0}^{p-1}\sum_{j\in\lambda_t} tA_{i,j} x_j.
    \end{align*}
    As $x$ ranges over $S_i$ (a vector space over $\bb Z_p$), the linear functional $\mu_i(x)$ could either be constantly $0$ or hit every value in $\bb Z_p$ equally many times. But because $\lambda$ has at least two blocks, this functional is not a constant multiple of $\sum_{j\in V_i} A_{i,j} x_j$, so has a kernel defining a hyperplane that is not equal to $S_i$. In other words, $\mu_i$ is not constantly $0$ for $x\in S_i$. Thus, $\mu_i$ must take each value equally many times, and so we get that
    $$\sum_{x\in S_i} \omega^{\mu_i(x)} = 0.$$
    Thus, $(2)=0$.
    
    Putting both pieces together, we get that
    $$\sum_{x\in S_i}\psi(a_{i,x}) = p^{-|V_i|} (p^{|V_i|} + 0) = 1$$

This completes the argument that $\psi:  \cl A(\mathrm{syncLCS}(A,b)) \to \bb C \Gamma(A,b)$ is a $*$-homomorphism.

Finally, we must show that $\overline{\psi}$ and $\tilde\phi$ are mutual inverses.

Let $\ell\in\{1,\dots,n\}$. Let $i$ be such that $\ell \in V_i$.
\begin{align*}
    \psi(\phi(g_\ell)) &= \psi\left(\sum_{x\in S_i} \omega^{x_\ell} a_{i,x}\right)\\
    &= \sum_{x\in S_i}\omega^{x_\ell} \prod_{j\in V_i} f_j(x_j)\\
    &= \sum_{x\in S_i}\omega^{x_\ell} f_\ell (x_\ell) \prod_{j\in V_i\backslash \{\ell\}} f_j(x_j)\\
    &= \sum_{x\in S_i}g_\ell f_\ell (x_\ell) \prod_{j\in V_i\backslash \{\ell\}} f_j(x_j)\\
    &= g_\ell \sum_{x\in S_i} \psi(a_{i,x})\\
    &= g_\ell
\end{align*}

where the fourth equality holds because of the following:

\begin{align*}
    \omega^s f_j(s) &= \frac{\omega^s}{p} \sum_{t=0}^{p-1}(\omega^{-s}g_j)^t\\
    &= \frac{1}{p} \sum_{t=0}^{p-1} \omega^{-s(t-1)}g_j^t\\
    &= \frac{g_j}{p} \sum_{t={-1}}^{p-2} \omega^{-st}g_j^t\\
    &= g_j f_j(s)
\end{align*}


If $q$ is the canonical quotient map $ \bb C \Gamma(A,b) \to \frac{\bb C \Gamma(A,b)}{\langle J-\omega 1 \rangle}$, then for all $j$,
$$\overline\psi\circ\tilde\phi(q(g_j)) = q(g_j).$$
Also, $\psi\circ\phi(J-\omega1) = \psi(0)= 0$. Thus, $$\overline\psi\circ\tilde\phi(q(J)) = q(J).$$ In summary, we get that $\overline\psi\circ\tilde\phi$ is the identity on $\frac{\bb C \Gamma(A,b)}{\langle J-\omega 1 \rangle}$.

In the other direction, let $i \in \{1,\dots,m\}$, and $y\in S_i$. Then we have 
\begin{align*}
    \phi(\psi(a_{i,y})) &= \phi\left(\prod_{j\in V_i} f_j (y_j)\right)\\
    &= \prod_{j\in V_i}\frac{1}{p} \sum_{t=0}^{p-1}\left(\omega^{-y_j}\phi(g_j)\right)^t\\
    &= \frac{1}{p^{|V_i|}} \prod_{j\in V_i}\sum_{t=0}^{p-1}\left( \omega^{-y_j}\sum_{x\in S_i} \omega^{x_j}a_{i,x}\right)^t\\
    &= \frac{1}{p^{|V_i|}} \prod_{j\in V_i}\sum_{x\in S_i}\left(\sum_{t=0}^{p-1} \omega^{(x_j - y_j)t}\right)a_{i,x}\\
    &= \frac{1}{p^{|V_i|}} \sum_{x\in S_i}\left(\prod_{j\in V_i}\sum_{t=0}^{p-1} \omega^{(x_j - y_j)t}\right)a_{i,x}\\
    &= \frac{1}{p^{|V_i|}}\prod_{j\in V_i}\sum_{t=0}^{p-1} \omega^{(y_j - y_j)t}a_{i,y}\\
    &= \frac{1}{p^{|V_i|}}\left(\prod_{j\in V_i} p \right)a_{i,y}\\
    &= a_{i,y}
\end{align*}

where the third-to-last equality holds because whenever there is some $x_j\neq y_j$, we get the $p$ distinct $p^\textrm{th}$ roots of unity which sum to $0$. Thus, the only $x$ that survives is $x=y$.

Note that $\tilde\phi$ is an induced homomorphism satisfying $\tilde\phi \circ q = \phi$. So $$\tilde\phi\circ \overline\psi = \tilde\phi\circ q\circ \psi= \phi\circ\psi = \mathrm{Id}.$$

\end{proof}

%% file: sec/graph_games.tex
So far, we have seen two equivalent algebraic structures modelling the game syncLCS. Now we will see that this structure connects syncLCS to another class of games, known as graph isomorphism games. 

We can associate a graph to the system $Ax=b$. This is the same $G_{A,b}$ defined in \cite{AMRSSV}, but now $A,b$ have entries in $\bb Z_p$.

\begin{defn}
    Let $G_{A,b}$ be the graph with vertex set $\{(i,x): i\in \{1,\dots,m\}, x\in S_i(A,b)\}$ and $(i,x) \sim (j,y)$ if there exists $k\in V_i \cap V_j$ with $x_k\neq y_k$.
\end{defn}
Recall that $S_i(A,b)$ is the set of solutions to row $i$ of the system of equations $Ax = b$, and $V_i$ is the support of row $i$ of $A$. We can think of edges in $G_{A,b}$ as capturing incompatibility of solutions to the various rows. 

We will now see that a game asking the players to construct an isomorphism from $G_{A,b}$ to $G_{A,0}$ is essentially the same as the game $\textrm{syncLCS}(A,b)$, which asks players to demonstrate a shared solution to the system $Ax=b$. This graph game is known more generally as $\textrm{Iso}(G,H)$ for two graphs $G$ and $H$. It was introduced in \cite{AMRSSV} but we will define it here for completeness.

\begin{defn}
Let $G,H$ be two graphs. The synchronous game $\textrm{Iso}(G,H)$ has input set $\cl{I} = V(G)\sqcup V(H)$, output set $\cl{O} = V(G)\sqcup V(H)$, and rule function $\lambda(x,y|v,w)=1$ exactly when these conditions are met:
\begin{itemize}
    \item $v$ and $x$ belong to different graphs,
    \item $w$ and $y$ belong to different graphs, and
    \item if $v,w$ are in the same graph, then their relationship in that graph (equal, adjacent, or distinct nonadjacent) must also be satisfied by $x,y$. 
\end{itemize}
\end{defn}

The following theorem was first stated as such in an early version of \cite{BCEHPSW} for the $p=2$ case, but the proof had an error, as there was a third game included in the equivalence which is now known to only satisfy a weaker form of equivalence to these two games. That third game is the graph homomorphism game from $K_m$ to $\overline{G_{A,b}}$. The result in \cite{BCEHPSW} was then weakened to say the three games are hereditarily $*$-equivalent, meaning that the hereditary quotients of the game algebras possess a pair of $*$-homomorphisms.

The following theorem generalizes that equivalence result to the $p>2$ setting, and drops the modifier ``hereditarily'' once again, but only for two of the three games originally considered.

\begin{thm}
\label{thm:sync-iso}
    Let $A$ be an $m\times n$ matrix over $\bb Z_p$ and let $b \in \bb Z_p^n$. Then the following two synchronous games are $*$-equivalent:
    \begin{enumerate}
        \item $\textrm{syncLCS}(A,b)$
        \item $\textrm{Iso}(G_{A,b},G_{A,0})$
    \end{enumerate}
\end{thm}

Before seeing the proof, it will be useful to learn a little about the game algebra of $\textrm{Iso}(G,H)$.

For graphs $G,H$, $\cl A(\textrm{Iso}(G,H))$ is generated by $e_{g,h}$ where $g\in V(G)$ and $h\in V(H)$. Looking at the definition of $\textrm{Iso}(G,H)$ and the definition of the game-algebra, it seems we may need the other three possibilities of which graphs the pair $(g,h)$ can live in. However, it can be shown that $e_{g,h} = 0$ whenever $g,h$ are vertices in the same graph. With some additional cleverness, it can also be shown that $e_{h,g} = e_{g,h}$, eliminating the need for both index orderings.

This means that in our case, $\cl A(\textrm{Iso}(G_{A,b},G_{A,0}))$ is generated by $e_{(i,x),(j,y)}$ where $i,j\in\{1,\dots,m\}, x\in S_i(A,b), y\in S_j(A,0)$. That is, each $(i,x)$ is a vertex of $G_{A,b}$ and $(j,y)$ is a vertex of $G_{A,0}$. The usual game algebra relations are in place: each generator must be a self-adjoint idempotent. Also for a fixed input, the sum over all outputs must be $1$. In that case, this gives us the following two relations:
\begin{align}
    \sum_{(i,x)} e_{(i,x),(j,y)} &= 1\\
    \sum_{(j,y)} e_{(i,x),(j,y)} &= 1
\end{align}
Furthermore, the product of two generators must be $0$ when the corresponding inputs and outputs don't satisfy the game's rules.
That means $e_{(i,x),(j,y)}e_{(i',x'),(j',y')}=0$ whenever the relationship of the vertices $(i,x), (i',x')$ in $G_{A,b}$ is the same as the relationship of the vertices $(j,y),(j',y')$ in $G_{A,0}$. In more detail,  $e_{(i,x),(j,y)}e_{(i',x'),(j',y')}=0$ whenever:
\begin{itemize}
    \item $(i,x)=(i',x')$ but $(j,y)\neq (j',y')$, or $(j,y)= (j',y')$ but $(i,x)\neq(i',x')$. That is, one pair of vertices is equal when the other pair is not equal.
    \item $(i,x)\sim(i',x')$ but $(j,y)\nsim (j',y')$, or $(j,y) \sim (j',y')$ but $(i,x)\nsim(i',x')$. That is, one pair of vertices is incompatible when the other pair is compatible.
\end{itemize}

\begin{proof}[Proof of Theorem \ref{thm:sync-iso}]


We claim this map is a unital $*$-homomorphism:
\begin{align*}
    \phi:\cl A(\textrm{Iso}(G_{A,b}, G_{A,0})) &\to \cl A(\textrm{syncLCS}(A,b))\\
    e_{(i,x),(j,y)}&\mapsto \delta_{ij}a_{i,x+y}
\end{align*}
    for all $i,j\in\{1,\dots,m\}$ and all $x\in S_i(A,b), y\in S_j(A,0)$.
Each generator maps to a self-adjoint idempotent. We must verify that $\phi$ preserves the other relations of $A(\textrm{Iso}(G_{A,b}, G_{A,0}))$. Unitality will follow from these.

First we show that the ``sum to one'' relations are preserved by $\phi$.

\begin{align*}
    \sum_{(i,x)} \phi(e_{(i,x),(j,y)}) &= \sum_{(i,x)}\delta_{ij}a_{i,x+y}\\
    &= \sum_{x\in S_j(A,b)} a_{j,x+y}\\
    &= \sum_{z\in S_j(A,b)} a_{j,z}\\
    &= 1
\end{align*}
The second-to-last step here is due to the fact that the set of solutions to $Ax=b$ is preserved by adding a homogeneous solution.
Similarly,
\begin{align*}
    \sum_{(j,y)} \phi(e_{(i,x),(j,y)}) &= \sum_{(j,y)}\delta_{ij}a_{i,x+y}\\
    &= \sum_{y\in S_i(A,0)} a_{i,x+y}\\
    &= \sum_{z\in S_j(A,b)} a_{j,z}\\
    &= 1
\end{align*}
The second-to-last step here is due to a similar fact: The set of solutions to $Ax=b$ can be constructed by starting with a particular solution and adding to it all homogeneous solutions.
From the above arguments, we get that $\phi$ sends $1$ to $1$. 

Now we show that the relations corresponding to the rules of the isomorphism game are preserved under $\phi$.

Assume $e_{(i,x),(j,y)}e_{(i',x'),(j',y')}=0$. Then we would like to show that $$\phi(e_{(i,x),(j,y)})\phi(e_{(i',x'),(j',y')}) = \delta_{ij}a_{i,x+y}\delta_{i'
j'}a_{i',x'+y'}$$
evaluates to $0$ as well.
Whenever $i\neq j$ or $i'\neq j'$, the above expression is $0$, so we may assume $i=j, i'=j'$, and we must show that $e_{(i,x),(i,y)}e_{(i',x'),(i',y')}= 0 \implies a_{i,x+y}a_{i',x'+y'}=0$. That is, we must show that $(x+y, x'+y')$ are incompatible solutions to rows $(i,i')$.

If $e_{(i,x),(i,y)} e_{(i',x'),(i',y')}=0$, there are several cases to consider.
\begin{enumerate}
    \item If $(i,x) = (i',x')$, we must have $(i,y) \neq (i',y')$. Because $i=i'$, we have that $y\neq y'$. Then, because $x=x'$, $x+y \neq x'+y'$, and so $(x+y, x'+y')$ are incompatible solutions to row $i$.
    \item Now if $(i,y) = (i',y')$, we must have  $(i,x) \neq (i',x')$. Exchanging $x$ and $y$ in the above argument shows that $(x+y, x'+y')$ are once again incompatible solutions to row $i$.
    \item If $(i,x) \sim (i',x')$, then there exists $k_0\in V_i \cap V_{i'}$ with $x_{k_0} \neq x_{k_0}'$. But also, we must have that $(i,y) \nsim (i',y')$, so for all $k\in V_i \cap V_{i'}$, $y_k = y_k'$. Thus,  $(x+y)_{k+0} \neq (x'+y')_{k_0}$, so $(x+y, x'+y')$ are incompatible solutions to rows $(i,i')$.
    \item If $(i,y) \sim (i',y')$, then exchanging $x$ and $y$ in the above argument resolves this last case.
\end{enumerate}
So, we have shown that in all cases, $\phi$ preserves the ``rules'' relations of the isomorphism game. We conclude that $\phi$ is a unital $*$-homomorphism $\cl A(\textrm{Iso}(G_{A,b}, G_{A,0})) \to \cl A(\textrm{syncLCS}(A,b))$.

In the other direction, let
\begin{align*}
    \psi:\cl A(\textrm{syncLCS}(A,b)) &\to \cl A(\textrm{Iso}(G_{A,b}, G_{A,0})) \\
    a_{i,x} &\mapsto \sum_{k=1}^m e_{(i,x),(k,0)}
\end{align*}
for all $i\in\{1,\dots,m\}, x\in S_i(A,b)$.
$\psi(a_{i,x})$ is self-adjoint and idempotent because $\{e_{(i,x),(k,0)}:k=1,\dots,m\}$ are a family of orthogonal projections.

When $(x,y)$ are incompatible solutions to rows $(i,j)$, we have $a_{i,x}a_{j,y} = 0$. But then $(i,x)\sim(j,y)$ in $G_{A,b}$, so
\begin{equation*}
    \psi(a_{i,x})\psi(a_{j,y}) = \sum_{k,\ell=1}^m e_{(i,x),(k,0)}e_{(j,y),(\ell,0)} = 0 
\end{equation*}
because each $(k,0)\nsim (\ell,0)$ in $G_{A,0}$ (they are either the same vertex, or distinct but compatible and therefore nonadjacent).

Finally, we must show that $\sum_{x\in S_i} \psi(a_{i,x}) = 1$. This piece of the proof relies on the fact that $\cl A(\textrm{Iso}(X,Y))$ is *-isomorphic to  $\cl O(G_X,G_Y)$ as defined in \cite{BCEHPSW}. In that paper, Theorem 4.7 tells us that whenever $\cl A(\textrm{Iso}(X,Y))$ is nonzero, it admits a faithful state. Therefore, $\cl A(\textrm{Iso}(X,Y))$ is hereditary. A *-algebra $\cl A$ is hereditary whenever $\sum_{i=1}^N x_i^*x_i = 0$ implies each $x_i = 0$, where $x_i \in \cl A$.

We follow exactly the steps taken in \cite{BCEHPSW} in the proof of Theorem 5.6.
Let $p_i:= \sum_{x\in S_i(A,b)} \sum_{k=1}^m e_{(i,x),(k,0)}$. We would like to show that $p_i = 1$. Note that $p_i$ is self-adjoint, being a sum of projections, and idempotent, because $e_{(i,x),(k,0)}e_{(i,y),(\ell,0)} = 0$ unless $x=y$ and $k=\ell$. Then, $q_i:= 1-p_i$ is also a self-adjoint idempotent. Thus,
\begin{align*}
    \sum_{i=1}^m q_i^*q_i &= \sum_{i=1}^m 1-p_i\\
    &= m\cdot 1 - \sum_{k=1}^m \sum_{i=1}^m \sum_{x\in S_i(A,b)}e_{(i,x),(k,0)}\\
    &= m\cdot 1 -\sum_{k=1}^m 1\\
    &= 0  
\end{align*}
so because the algebra is hereditary, we have that each $q_i=0$. Thus, each $p_i = 1$. We can also conclude from this that $\psi$ is unital.

$\psi$ is a unital $*$-homomorphism $\cl A(\textrm{syncLCS}(A,b)) \to \cl A(\textrm{Iso}(G_{A,b}, G_{A,0}))$. We have exhibited a unital $*$-homomorphism in both directions, meaning that the games $\textrm{syncLCS}(A,b)$ and $\textrm{Iso}(G_{A,b}, G_{A,0})$ are $*$-equivalent.
\end{proof}

%% file: syncLCS.bbl
\newcommand{\etalchar}[1]{$^{#1}$}
\begin{thebibliography}{AMR{\etalchar{+}}19}

\bibitem[AMR{\etalchar{+}}19]{AMRSSV}
Albert Atserias, Laura Man{\v{c}}inska, David~E Roberson, Robert
  {\v{S}}{\'a}mal, Simone Severini, and Antonios Varvitsiotis.
\newblock Quantum and non-signalling graph isomorphisms.
\newblock {\em Journal of Combinatorial Theory, Series B}, 136:289--328, 2019.

\bibitem[BCE{\etalchar{+}}19]{BCEHPSW}
Michael Brannan, Alexandru Chirvasitu, Kari Eifler, Samuel Harris, Vern
  Paulsen, Xiaoyu Su, and Mateusz Wasilewski.
\newblock Bigalois extensions and the graph isomorphism game.
\newblock {\em Communications in Mathematical Physics}, 375(3):1777–1809, Sep
  2019.

\bibitem[CLS17]{CLS}
Richard Cleve, Li~Liu, and William Slofstra.
\newblock Perfect commuting-operator strategies for linear system games.
\newblock {\em Journal of Mathematical Physics}, 58(1):012202, Jan 2017.

\bibitem[CM14]{CM}
Richard Cleve and Rajat Mittal.
\newblock Characterization of binary constraint system games.
\newblock In {\em International Colloquium on Automata, Languages, and
  Programming}, pages 320--331. Springer, 2014.

\bibitem[HMPS17]{HMPS}
William Helton, Kyle~P. Meyer, Vern~I. Paulsen, and Matthew Satriano.
\newblock Algebras, synchronous games and chromatic numbers of graphs, 2017.

\bibitem[JNV{\etalchar{+}}20]{JNVWY}
Zhengfeng Ji, Anand Natarajan, Thomas Vidick, John Wright, and Henry Yuen.
\newblock {MIP*=RE}, 2020.

\bibitem[KPS18]{KPS}
Se-Jin Kim, Vern Paulsen, and Christopher Schafhauser.
\newblock A synchronous game for binary constraint systems.
\newblock {\em Journal of Mathematical Physics}, 59(3):032201, Mar 2018.

\bibitem[LMP{\etalchar{+}}20]{LMPR}
M.~Lupini, L.~Mančinska, V.~I. Paulsen, D.~E. Roberson, G.~Scarpa,
  S.~Severini, I.~G. Todorov, and A.~Winter.
\newblock Perfect strategies for non-local games.
\newblock {\em Mathematical Physics, Analysis and Geometry}, 23(1), Feb 2020.

\bibitem[OP15]{OP}
Carlos~M. Ortiz and Vern~I. Paulsen.
\newblock Quantum graph homomorphisms via operator systems, 2015.

\bibitem[PSS{\etalchar{+}}16]{PSSTW}
Vern~I. Paulsen, Simone Severini, Daniel Stahlke, Ivan~G. Todorov, and Andreas
  Winter.
\newblock Estimating quantum chromatic numbers.
\newblock {\em Journal of Functional Analysis}, 270(6):2188–2222, Mar 2016.

\end{thebibliography}
